\documentclass[aps, reprint]{revtex4-1}
\pdfoutput=1
\usepackage{amsmath}
\usepackage{amssymb}
\usepackage{graphicx}
\usepackage{enumerate}
\usepackage{enumitem}

\providecommand{\U}[1]{\protect\rule{.1in}{.1in}}
\newtheorem{theorem}{Theorem}

\newtheorem{corollary}[theorem]{Corollary}

\newtheorem{definition}[theorem]{Definition}

\newtheorem{lemma}[theorem]{Lemma}

\newtheorem{proposition}[theorem]{Proposition}

\newenvironment{proof}[1][Proof]{\noindent\textbf{#1.} }{\ \rule{0.5em}{0.5em}}
\begin{document}
\title{Local Distinguishability of Generic Unentangled Orthonormal Bases}
\author{Ji\v{r}\'{\i} Lebl}
\thanks{Ji\v{r}\'{\i} Lebl was partially supported by NSF grant DMS-1362337 and
Oklahoma State University's DIG and ASR grants.}
\email{lebl@math.okstate.edu}
\affiliation{Department of Mathematics,  Oklahoma State University, 
Stillwater,  OK 74078,  USA}

\author{Asif Shakeel}
\thanks{Asif Shakeel was partially supported by NSF award PHY-0955518.}
\email{ashakeel@ucsd.edu}
\affiliation{Department of Mathematics,  University of California,  San Diego, 
 La Jolla,  CA 92093-0112,  USA}

\author{Nolan Wallach}
\thanks{Nolan Wallach was partially supported by NSF grant DMS-0963035.}
\email{nwallach@ucsd.edu}
\affiliation{Department of Mathematics,  University of California,  San Diego, 
 La Jolla,  CA 92093-0112,  USA}

\date{January 4, 2016}

\begin{abstract}
An orthonormal basis consisting of unentangled  (pure tensor) elements in  a
tensor product of Hilbert spaces  is an  Unentangled Orthogonal  Basis
(UOB).  In general,  for $n$ qubits,  we prove that in its  natural
structure as a real variety,  the space of UOB is a bouquet of products of
Riemann spheres parametrized by a class of edge colorings of hypercubes. Its
irreducible components of maximum dimension are products of $2^n-1$
two-spheres. Using a theorem of Walgate and Hardy,  we observe that the  UOB
whose elements are distinguishable  by local operations and classical
communication  (called {\it locally distinguishable} or {\it LOCC
distinguishable} UOB)  are exactly those  in the maximum dimensional
components. Bennett et al,  in their  in-depth study of quantum nonlocality
without entanglement, include a specific $3$ qubit example UOB which is not
LOCC   distinguishable; we construct certain generalized counterparts of
this UOB in $n$ qubits.
\end{abstract}

\maketitle

\section{Introduction} \label{sec:intro}
Quantum nonlocality through entanglement plays a  key role as a resource in
quantum  teleportation,  cryptography and error-correcting codes. There
exists,  however,  another  nonlocal phenomenon: quantum nonlocality without
entanglement, studied at length by Bennett, DiVincenzo, Fuchs, Mor, Rains,
Shor, Smolin, and Wootters in~\cite{bdf:qnwe}.   Locality in this sense
refers to the elements of  an unentangled orthogonal basis (UOB) being
distinguishable by a protocol using only  local operations by participants
(each holding one tensor factor) and classical communication (LOCC) among
them, hence such a UOB is {\it locally  distinguishable} or {\it LOCC
distinguishable}. In~\cite{bdf:qnwe}, the authors  provide examples of sets
of unentangled states that are not LOCC (locally)  distinguishable and
therefore exhibit nonlocality, give   measurement protocols for  their
optimal distinguishability,  state preparation protocols to obtain them,
relation to quantum cryptography, and measures to quantify their
nonlocality. This form of nonlocality  is connected with construction of
entangled states~\cite{dms:upbbe},  but  stands on its own as
well~\cite{bdf:qnwe,  scj:lsmco,  wh:nadbs}.   Under protocols in which
various parties can only measure  their own systems (local measurements)
and classically communicate,  distinguishing among certain unentangled
states is impossible. Thus,  such states encode quantum nonlocality.  It  is
demonstrably useful in quantum key distribution (QKD), as first observed by
Goldenberg and Vaidman~\cite{vg:qcbos}, and since used in other schemes for
secure communication~\cite{gls:nqkds,h:qkdosbc}.

The converse is to find   sets of states that are identifiable through LOCC. A significant  body of work is involved with  the criteria for recognizing and constructing  such states,  particularly pertaining  the  Unextendible Product Basis (UPB)~\cite{bdf:qnwe,  scj:lsmco,  dms:upbbe, cj:msupbbc}, but there have also been some  characterizations of   Unentangled Orthogonal  Basis  (UOB) and of the form of nonlocality in unentangled settings in higher dimensions~\cite{zgt:nopbqs,fs:cliops}. In this paper we analyze the set of orthonormal bases consisting of unentangled states (UOB) in $n$ qubits. We show that in the natural structure of  UOB as an algebraic variety over $\mathbb{R}$,  the ones that can be distinguished by LOCC are precisely those belonging to  the irreducible components of highest dimension.

The organization of this paper is as follows. In Section~\ref{sec:background}, we motivate and give an initial set of definitions that connect the orthognality condition of a UOB to a set of colorings of a hypercube that we call  {\it admissible}. In Section~\ref{sec:mainsec}, we prove the main result that relates maximum number of colors in an admissible coloring of a hypercube to the maximum dimensional component in the set of UOB,   and describe its implications. Section~\ref{sec:loccdist} discusses the LOCC distinguishability of the UOB, and through the theorem of Walgate and Hardy~\cite{wh:nadbs}, shows that the maximum dimensional component is the unique such set. In Section~\ref{sec:nonlocc} we construct UOB not distinguishable by LOCC, hence exhibiting the said nonlocality, and important from the secure communication point of view.  Section~\ref{sec:disc} discusses the directions for further research and certain open questions.

\section{Admissible Colorings and Orthogonality} \label{sec:background}

We denote by $\mathcal{H}_{n}$ the space of $n$ qubits that is
${\otimes}^{n}\mathbb{C}^{2}$ with the tensor product Hilbert
structure,  $\left\langle \ldots|\ldots\right\rangle $. A state in $\mathcal{H}_{n}$
is called a product state or unentangled state if it is a tensor product
of unit vectors in each $\mathbb{C}^{2}$. We note that two product states $v_{1}\otimes
v_{2}\otimes\cdots\otimes v_{n}$ and $w_{1}\otimes w_{2}\otimes\cdots\otimes
w_{n}$ satisfy
\begin{equation*}
\left\langle v_{1}\otimes v_{2}\otimes\cdots\otimes v_{n}|w_{1}\otimes
w_{2}\otimes\cdots\otimes w_{n}\right\rangle =0
\end{equation*}
if and only if there is at least one $i$ with $\left\langle v_{i}
|w_{i}\right\rangle =0$. Since states are determined up to phase,   to think
about them unambiguously we must consider them to be elements of the
corresponding projective space. If $z\in\mathbb{C}^{2}$ is non-zero
then we  assign to $z$ the complex line $[z]$ through $0$ and $z$. The
totality of elements $[z]$,  $z\in\mathbb{C}^{2}-\{0\}$ is denoted (as
usual) as $\mathbb{P}^{1}$ (one-dimensional projective space over $\mathbb{C}
$). Up to phase,  the element $v_{1}\otimes v_{2}\otimes\cdots\otimes v_{n}$ 
is  considered to be $[v_{1}]\otimes\lbrack v_{2}]\otimes\cdots
\otimes\lbrack v_{n}]$. On $\mathbb{P}^{1}$ we define a real analytic fixed
point free involution:
\begin{equation*}
\lbrack v]\longmapsto\lbrack\hat{v}], 
\end{equation*}
which assigns to $[v]$ the line $[\hat{v}]$ perpendicular to it (i.e.
$\left\langle v|\hat{v}\right\rangle =0$). If $S$ is a subset of $\mathbb{P}^1$ then $\hat{S}$ denotes the set of $[\hat{s}]$ for $[s]$ in $S$.

Our first goal is to turn the determination of all UOB into a
combinatorial problem on the hypercube $Q_{n}$. We think of the vertices of
the hypercube as the vectors in $\mathbb{R}^{n}$ with coordinates in the set
$\{0, 1\}^{n}$,  and consider this to be     binary expansions of  numbers
$0, 1, \ldots, 2^{n-1}$. We also view  $Q_{n}$ as a graph with vertices
$0, 1, \ldots, 2^{n-1}$; its edges are the pairs of numbers whose binary
expansions differ in exactly one digit (i.e. pairs with Hamming distance $1$).

Let $u_{0}, u_{1}, \ldots, u_{2^{n}-1}$ be a UOB,  and write  its states as
\begin{equation*}
\lbrack u_{j}]=[u_{1j}]\otimes\lbrack u_{2j}]\otimes\cdots\otimes\lbrack
u_{nj}].
\end{equation*}
As observed above,  if $i\neq j$ then at least one pair $\{[u_{ki}], [u_{kj}]\}$
must be of the form $\{[v], [\hat{v}]\}$. We consider the subset of
$\mathbb{P}^{1}$ that is the set $\mathcal{T=}\{[u_{kj}
]|k=1, \ldots, n, j=0, \ldots, 2^{n}-1\}$. We divide $\mathcal{T}$ into two disjoint
pieces $\mathcal{T}_{0}$ and $\mathcal{T}_{1}$ such that $\mathcal{\hat{T}
}_{i}\cap\mathcal{T}_{i}=\emptyset$,  $i=0, 1$. This implies that if
$[t]\in\mathcal{T}_{0}$ and if $[\hat{t}]\in\mathcal{T}$ then $[\hat{t}
]\in\mathcal{T}_{1}$ and vice-versa. To each $[u_{j}]$ we assign a vector
$s_{j}\in\mathbb{R}^{n}$ such that its $k$--th coordinate is $0$ if
$[u_{kj}]\in\mathcal{T}_{0}$ or $1$ if $[u_{kj}]\in\mathcal{T}_{1}.$ We note
that if we assign to $s_{j}$ the corresponding element
\begin{equation*}
\left [  s_{j}\right\rangle =\left [  s_{1j}s_{2j}\ldots s_{nj}\right\rangle, 
\end{equation*}
then by its very definition $\{\left [  s_{j}\right\rangle |j=0, \ldots , 2^{n}-1\}$
is an orthonormal set. This implies that the two sets $\mathcal{T}_{0}$ and
$\mathcal{T}_{1}$ each consist of exactly half of the elements of $\mathcal{T}
$ and that $\mathcal{T}_{1}=\mathcal{\hat{T}}_{0}$. Reordering $\mathcal{T}$, 
let $\mathcal{T}_{0}=\{t_{1}, \ldots , t_{r}\}$,  such that $s_{j}$ is just
the binary expansion of $j$. Assume a palette of colors
$c_{1}, c_{2}, \ldots, c_{r}, \ldots$ is available. From this palette,  we  assign to each vertex of
$Q_{n}$ an n--tuple of colors taken from $c_{j}$ with $1\leq j\leq r$,  such that if the $i$--factor
of $u_{j}$ is $t_{j}$ or $\hat{t}_{j}$,  we assign to it the color $c_{j}$. This
is equivalent to coloring the edges of $Q_{n}$. Indeed,  let $a\frac{\qquad}
{{}}b$ be an edge,  so $a$ and $b$ differ in exactly one component,  which,  by orthonormality,  has the same  color in both $a$ and $b$. We give the edge $a\frac{\qquad}
{{}}b$ 
that color. Conversely,  given  an edge-coloring of $Q_{n}$,  we can assign  an $n$-tuple of 
colors to each vertex as follows. For the  vertex $a$ and component $i$,   let $a^{i}$ be the unique vertex with all its components the same as those of $a$ except for the $i$--th which is
opposite. We assign the $i$--th component of vertex $a$ the color of edge $a\frac{\qquad}{{}}a^{i}$.

\begin{definition}
A coloring of $Q_{n}$ is said to be admissible if for every pair of vertices
there is a component,  $i$,  so that one vertex has a $0$ in the $i$--th position
and the other has a $1$ and  both are assigned
the same color in that position.
\end{definition}

If we have a coloring of $Q_{n}$ with colors $c_{1}, \ldots, c_{k}$ and
$[u_{1}], \ldots, [u_{k}]$ are elements of \thinspace$\mathbb{P}^{1}$ then we
assign to each vertex $s=s_{1}s_{2}\ldots s_{n}$ a product state (up to phase):
 if the $i$--component has color $c_{r}$ and $s_{i}=0$ then put
$[u_{r}]$ in the $i$--th position; if $s_{i}=1$ put $[\hat{u}_{r}]$ in the
$i$--th position. For example,  for  $n=3$ we have the admissible coloring
given in FIG.~\ref{fig1}.
\begin{figure}[ht]
\includegraphics[width=1.5in]{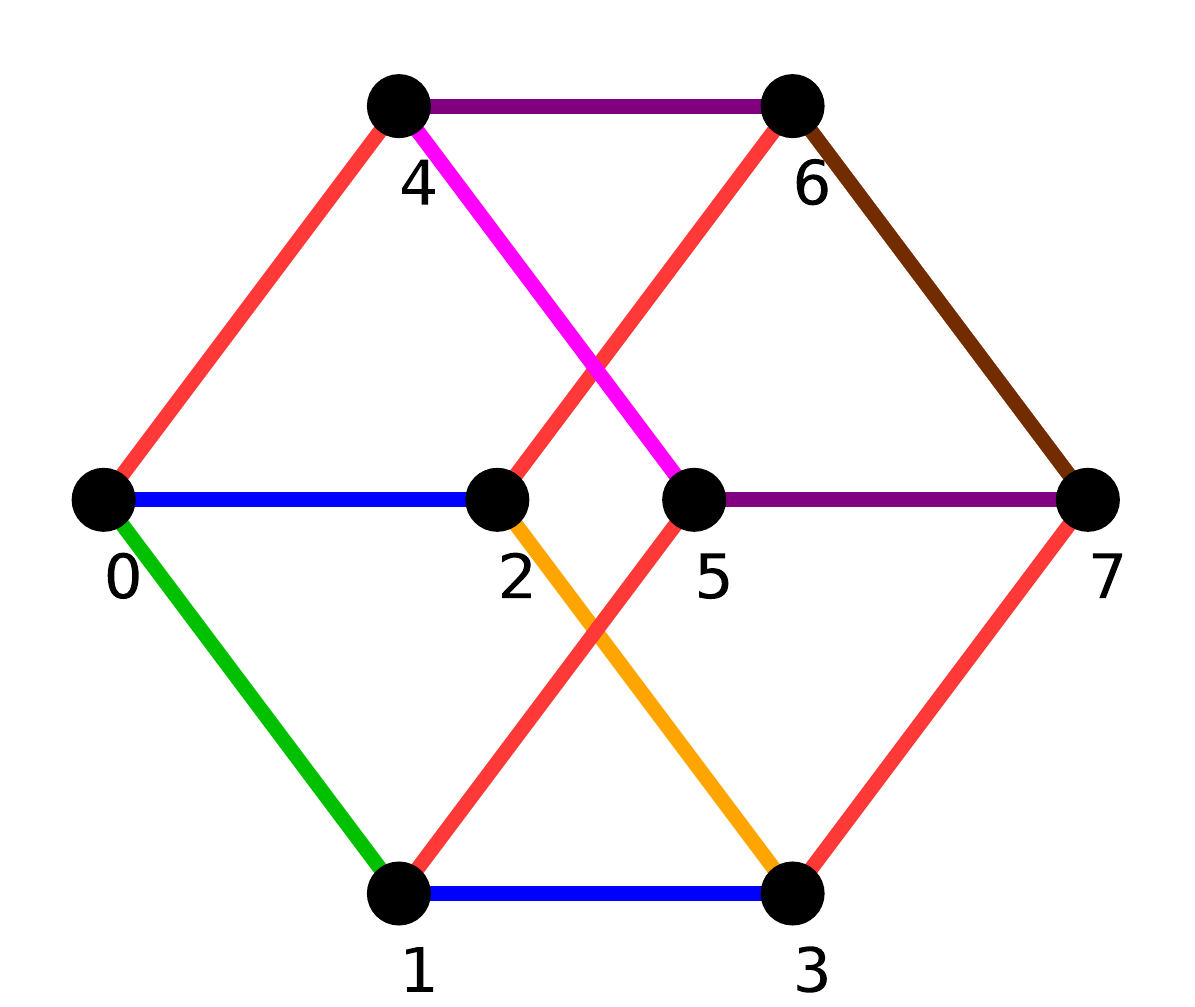}
\caption{\label{fig1} (Color online)
Admissible coloring with 7 colors. The edges
0-4, 1-5, 3-7, 2-6 are red, the edges 0-2, 1-3 are blue, the edges
4-6, 5-7 are violet, the edges 0-1, 4-5, 6-7, and 2-3 are respectively
green, purple, brown, and orange.}
\end{figure}

Here $c_{1}=\mathrm{green}$,  $c_{2}=\mathrm{blue}$,  $c_{3}=\mathrm{red}
$,  $c_{4}=\mathrm{orange}$,  $c_{5}=\mathrm{purple}$,
$c_{6}=\mathrm{violet}$ and
$c_{7}=\mathrm{brown}$. The procedure assigns the UOB:
\begin{equation} \label{basisvec}
\begin{aligned}
&
[u_{3}]\otimes[u_{2}]\otimes[u_{1}], \quad
[u_{3}]\otimes[u_{2}]\otimes[\hat{u}_{1}], 
\\
&
[u_{3}]\otimes[\hat{u}_{2}]\otimes[u_{4}], \quad
[u_{3}]\otimes[\hat{u}_{2}]\otimes[\hat{u}_{4}], 
\\
&
[\hat{u}_{3}]\otimes[ u_{5}]\otimes[u_{6}], \quad
[\hat{u}_{3}]\otimes[ u_{5}]\otimes[\hat{u}_{6}], 
\\
&
[\hat{u}_{3}]\otimes[\hat{u}_{5}]\otimes[u_{7}], \quad
[\hat{u}_{3}]\otimes[\hat{u}_{5}]\otimes[\hat{u}_{7}].
\end{aligned}
\end{equation}
We give the set of UOB of $\mathcal{H}_{n}$,  $\mathcal{U}_{n}$,  its subspace
topology in the set of $2^{n}$--tuples of elements of the projective space on
$\mathcal{H}_{n}$,  $\mathbb{P (}\mathcal{H}_{n})$.

\section{Maximum Dimensional Component and Maximal Colorings} \label{sec:mainsec}

In this section, we obtain an interesting connection between  the admissible colorings and the maximum dimensional component. This comes about through an elegant structure of the admissible colorings when viewed as a combinatorial forest.

\begin{proposition}
Fix a pallette of colors $c_{1}, \ldots, c_{k}, \ldots$ To each admissible coloring, 
$C$,  of $Q_{n}$ with $k$ colors the procedure above yields an injective, 
continuous mapping
\begin{equation*}
\Phi_{C} \colon \left (  \mathbb{P}^{1}\right)  ^{k}\rightarrow\mathcal{U}_{n}.
\end{equation*}
The union of the images of  $\Phi_{C}$ running through all admissible
colorings is all of $\mathcal{U}_{n}$.
\end{proposition}

For each coloring $C$ the map $\Phi_{C}$ is a homeomorphism onto
its image. Thus $\mathcal{U}_{n}$ is a finite union of smooth manifolds
diffeomorphic with $\left (  \mathbb{P}^{1}\right)  ^{k}$ for $k$ running
through the cardinalities of admissible colorings of $Q_{n}$. We introduce a partial order on the set of colorings of $Q_n$.

\begin{definition}
If $C_{1}, C_{2}$ are colorings of $Q_{n}$ then $C_{1}\prec C_{2}$ 
if the colors used in $C_{1}$ form a subset,  $S$,  of
those used in $C_{2}$ and the set of edges that were colored in $C_{2}$ by color $c \notin S$ all have their color replaced by a color in $S$.

\end{definition}

\begin{lemma}
Up to changing the names of the admissible colors $C_{1}\prec C_{2}$ if and only if the
image of $\Phi_{C_{1}}$ is contained in that of $\Phi_{C_{2}}$.
\end{lemma}

We  make some observations about this ordering. If $C$ is a coloring of
$Q_{n}$ let $C (i)$ denote the colors of the edges with vertices that differ in
the $i$--th position. We  change the colors of each $C (i)$ so that $C (i)\cap C (j)=\emptyset$ if $i\neq j$. Thus in a maximal coloring every vertex has
$n$ distinct colors. There is a unique minimal coloring (up
to changing the names of the colors): the coloring with one color. This
coloring yields the tensor product of the standard orthogonal bases of
$\mathbb{C}^{2}$.

We will see in Theorem~\ref{main}  below that the admissible coloring of $Q_{3}$ above is maximal and has the maximum
number of colors,  $7$.  This
implies that $\mathcal{U}_{3}$ can be thought of as a bouquet of some fourteen
dimensional real manifolds and some lower dimensional ones corresponding to maximal
colorings with less than $7$ colors.  In FIG.~\ref{fig2} is an example of a maximal coloring
of $Q_{3}$ with $6$ colors.
\begin{figure}[ht]
\includegraphics[width=1.5in]{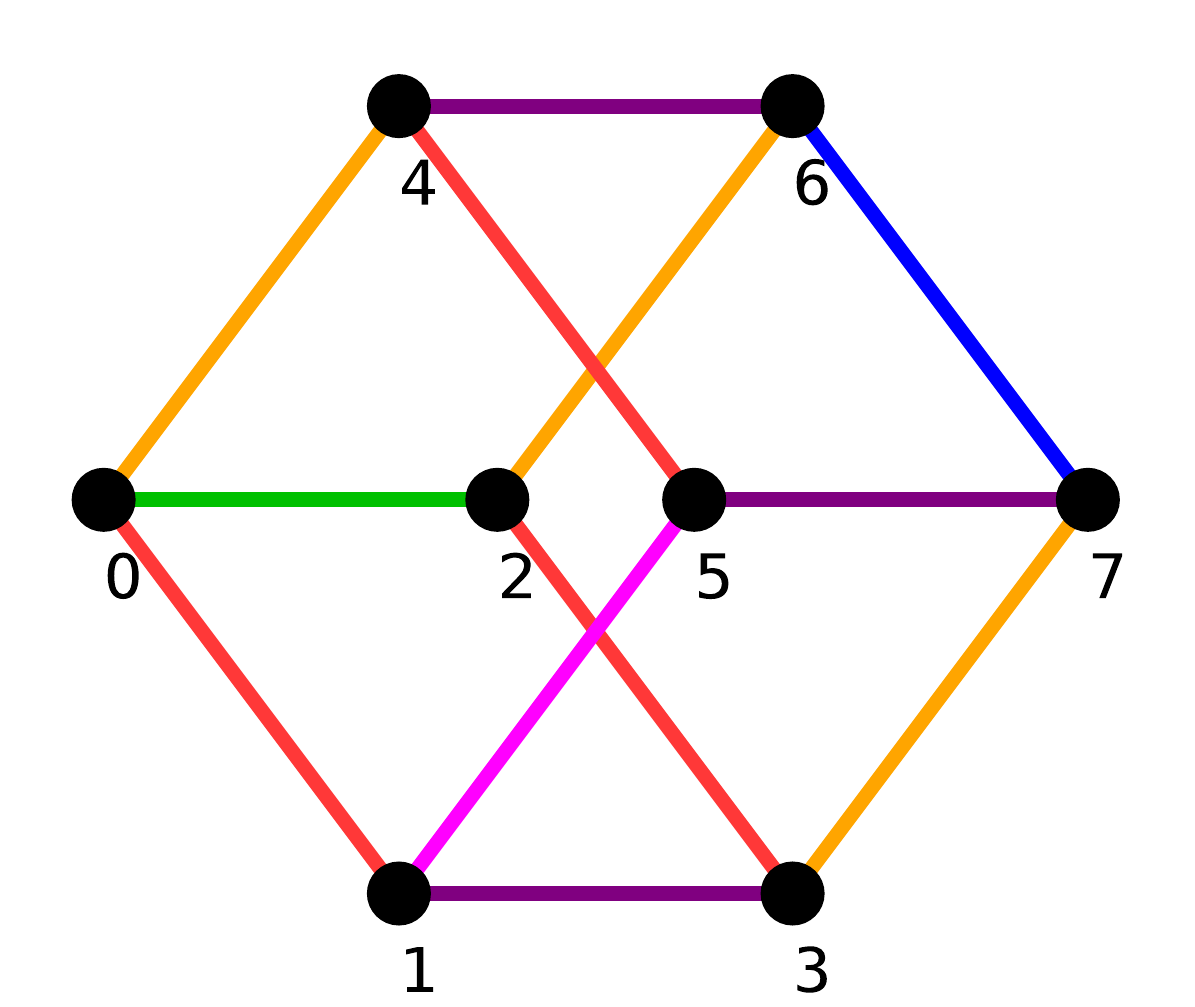}
\caption{\label{fig2} (Color online)
Admissible coloring with 6 colors.
The edges 0-4, 3-7, 2-6 are orange, the edges 0-1, 4-5, 2-3 are red,
the edges 4-6, 5-7, 1-3 are violet, and the edges
1-5, 0-2, and 6-7 are respectively purple, green, and blue.}
\end{figure}

The figure corresponds to the UOB:
\begin{equation*}
\begin{aligned}
&
[u_{3}]\otimes [u_{2}]\otimes [u_{1}], \quad
[u_{5}]\otimes [u_{4}]\otimes [\hat{u}_{1}], 
\\
&
[u_{3}]\otimes[\hat{u}_{2}]\otimes[u_{1}], \quad
[u_{3}]\otimes[\hat{u}_{4}]\otimes[\hat{u}_{1}], 
\\
&
[\hat{u}_{3}]\otimes[u_{4}]\otimes[u_{1}], \quad
[\hat{u}_{5}]\otimes[u_{4}]\otimes[\hat{u}_{1}], 
\\
&
[\hat{u}_{3}]\otimes[\hat{u}_{4}]\otimes[u_{6}], \quad
[\hat{u}_{3}]\otimes[\hat{u}_{4}]\otimes[\hat{u}_{6}].
\end{aligned}
\end{equation*}
With specific choices of $u_1,\ldots,u_6$, this example appears in~\cite{bdf:qnwe}.

In preparation for our main theorem we give a recursive algorithm for
admissibly coloring $Q_{n}$ with $2^{n}-1$ colors,  which the theorem asserts is
the maximum number. Also Theorem~\ref{theorem:maxcolors} implies this is the only way,  up to
permuting indices,  to color $Q_{n}$ admissibly with $2^{n}-1$ colors. 

\begin{lemma}\label{algoritheorem}
Let $C_{0}$ and $C_{1}$ be admissible colorings of $Q_{n-1}$. Writing
$Q_{n}$ as $0\times Q_{n-1}\cup1\times Q_{n-1}$ and choosing a new color $c$ then we color $Q_n$ as follows:
all first coordinates are colored with color $c$ if the first index is $0$ (respectively $1$) then the rest of the indices are colored as in 
$C_{0}$ (resp. $C_{1}$). This recipe yields an
admissible coloring. In particular,  if $C_{0}$ and $C_{1}$ both use 
$2^{n-1}-1$ colors without any repetitions between the colors,  then the number
of colors is $2^{n}-1$ for the coloring of $Q_{n}$.
\end{lemma}

In FIG.~\ref{fig3} is an example of this method for $Q_{5}$ (it uses the algorithm
starting with the $Q_{3}$ example above with $7$ colors to get a $Q_{4}$
coloring with $15$ colors and then another application to get $31$ colors).
\begin{figure}[ht]
\includegraphics[width=1.5in]{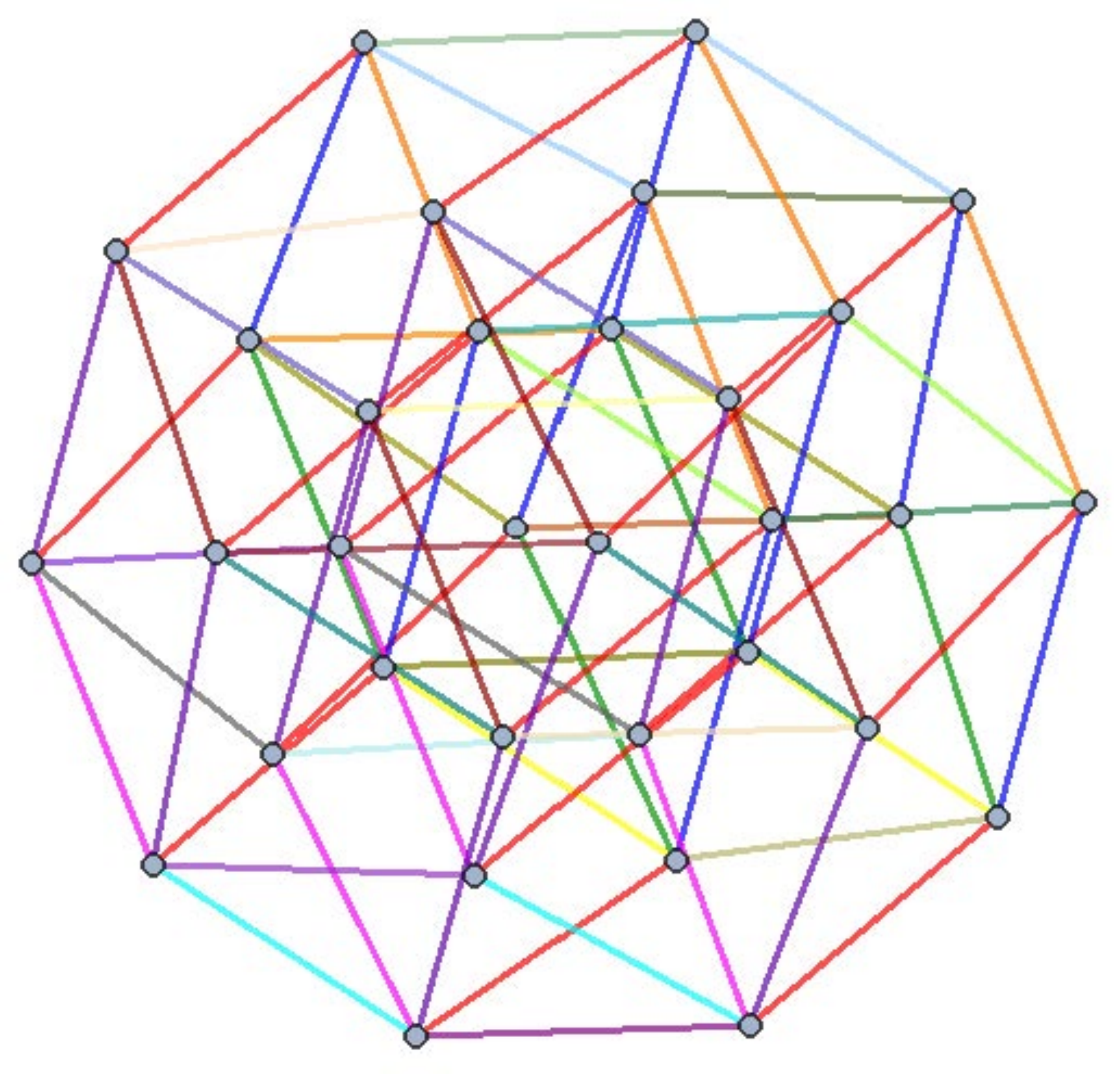}
\caption{\label{fig3} (Color online) Admissible coloring of $Q_5$ with 31 colors.}
\end{figure}

In the proof of the following result we will only use the admissibility of
every 2-face of an admissible coloring.

\begin{theorem}\label{main}
\begin{enumerate}[label= (\roman{*})] 
\item \label{main1} Let $Q_{n}$ be admissibly colored. Then there exists a subforest $F$ (i.e. a
subgraph with no circuits) of $Q_{n}$ that has edges of every possible color
in $Q_{n}$.
\item \label{main2}The maximum number of colors in an admissible coloring of $Q_{n}$ is $2^{n}-1
$.
\item \label{main3} $Q_{n}$ is admissibly colored with $2^{n}-1$ colors if and only if some
forest in $Q_{n}$ containing all of its colors each exactly once is a tree
that contains all  the vertices of $Q_{n}.$
\item \label{main4} If $Q_n$ is admissibly colored with $2^n-1$ colors then every subcube $Q_m$ where $m<n$ is also admissibly colored with 
$2^{m}-1$ colors.
\end{enumerate}
\end{theorem}

\begin{proof}
We first show how one can derive~\ref{main2} and~\ref{main3} from~\ref{main1}. To prove~\ref{main2},  we note that if a forest consists of $k$
disjoint trees and $m$ vertices then the number of edges is at most $m-k$.
Thus if $F$ is the forest asserted in~\ref{main1} then $m\leq2^{n}$.
As the number of colors is at most the number of its edges,  we have that  the number of
colors is at most $2^{n}-k$,  with $k$ the number of connected components
 (disjoint trees). This proves~\ref{main2}.

To prove~\ref{main3},  consider $F$,   a subforest of $Q_{n}$ containing $2^{n}-1$
edges. Then it must contain at least $2^{n}$ vertices and the number of
connected components is $1$. If $Q_{n}$ is admissibly colored and if $F$ is a
tree containing all of its colors each exactly once and all of the vertices of
$Q_{n}$ then since the number of edges is $2^{n}-1$,  that must be the number of colors.

We now prove~\ref{main1} by induction on $n$. If $n=1, 2$,  the result is
obvious. So we assume~\ref{main1} for $n-1\geq2$ and prove the result for $n$. Let
$Q_{n}^{j}$ be the set of elements of $Q_{n}$ with first coordinate $j$ with
$j=0$ or $1$. We take each to be an $n-1$ subcube and give each the coloring
that it inherits from $Q_{n}$. The inductive hypothesis implies that for each of 
these cubes there is respectively a sub-forest $F\subset Q_{n}^{0}$ and
$G\subset Q_{n}^{1}$ as in~\ref{main1}. From $G$ we delete all the edges with colors
that are in $F$. We now take $H$ to be $F\cup G$ with a subset of edges not in
the $Q_{n}^{j}$ (we call such edges vertical) adjoined that contain all of
the colors of $Q_{n}$ not contained in $F\cup G$ each exactly once. If we show
that $H$ has no cycles then~\ref{main1} is proved. Suppose on the contrary there is
a cycle in $H$. Then it cannot stay in $F$ and verticle edges or in $G$ and
vertical edges. Thus we may assume that it starts in $F$ at $p_{1}$
immediately goes vertical along $v_{1}$ then passes through $q_{1}
, q_{2}, \ldots, q_{k}$ \ in $G$ and then goes vertical along the edge $v_{2}$ which
connects to $q\in F$ . The circuit may not be as yet closed but we now show
that this is enough for a contradiction. In fact,  we show that $v_{1}$ and
$v_{2}$ must have the same color. Indeed,  consider the following diagram:
\begin{equation*}
\begin{array}
[c]{ccccccccccc}
& q_{1} & \rightarrow & q_{2} & \rightarrow & q_{3} & \cdots & \cdots &
q_{k-1} & \rightarrow & q_{k}\\
v_{1} & \uparrow & w_{1} & \uparrow & w_{2} & \uparrow & \cdots & w_{k-2} &
\uparrow & v_{2} & \uparrow\\
& p_{1} & \rightarrow & p_{2} & \rightarrow & p_{3} & \cdots & \cdots &
p_{k-1} & \rightarrow & q
\end{array}
.
\end{equation*}
In this diagram only the $q_{i}, p_{1}, q_{1}$ are guaranteed to be vertices in
$H$ and only $v_{1}$ and $v_{2}$ are vertical edges in $H$. However,  each of
the
\begin{equation*}
\begin{array}
[c]{cccc}
q_{i} & \rightarrow & q_{i+1} & \\
\uparrow & w_{i} & \uparrow & w_{i+1}\\
p_{i} & \rightarrow & p_{i_{+1}} &
\end{array}
\end{equation*}
is a $2$ dimensional subcube of $Q_{n}$. Since the edge $q_{i}
\rightarrow q_{i+1}$ is in $G$ and $p_{i}\rightarrow p_{i+1}$ is an edge of
$Q_{n}^{0}$,  the two edges have different colors. this implies that $w_{i}$ and
$w_{i+1}$have the same color (by admissibility). The argument applies to the
first and last square also so we see that $v_{1}$ and $v_{2}$ have the same
color contrary to the choice of edges to include.

Before we prove~\ref{main4} we recall a property of the forest $T$ that was found in
the proof of~\ref{main3}. There is no path in $T$ that starts in $F$ continues in $G$
and returns to $F$. We now prove~\ref{main4}.  We note that it is enough to prove this for codimension one
subcubes with the inherited coloring. If we choose one such subcube we  rotate it
so that it is $Q_{n}^{0}$. We now consider the forests $T$ and $F$. Since
$Q_{n}$ has $2^{n}-1$ colors $T$ must be connected. According to~\ref{main3} we will be
done if we show that $F$ is connected. To prove this we consider $x, y$
vertices in $F$. Since $T$ is connected there must be a path from $x$ to $y$
in $T.$ This path cannot leave $F$ and return to $F$. Thus it
stays in $F.$
\end{proof}

\begin{theorem} \label{theorem:maxcolors}
Let $Q_{n}$ be admissibly colored with $2^{n}-1$ distinct colors. Then there exists a direction for which all $2^{n-1}$ edges in that direction have the same color.
\end{theorem}

\begin{proof}
We first note that the theorem can be proved directly for $n=2, 3$. We also
observe that if $Q_{3}$ is colored admissibly with $7$ colors then if $3$ out
of $4$ of the edges in the same direction have the same color then so does the
fourth. The proof is by induction. Suppose $n\geq4$ and the lemma is true for
$Q_{n-1}$. We suppose that we have a maximal coloring of $Q_{n}$ with $2^{n}-1$ distinct colors. As before, 
let us split the $Q_{n}$ into two $n-1$ dimensional subcubes,  the top and the
bottom. Let us call them $Q^{ (0)}$ for the bottom and $Q^{ (1)}$ for the top.
The edges between them we call vertical. If all the vertical edges are of the
same color,  we are done. So suppose that there are at least two distinct
colors on the vertical edges. Let us call the vertical direction the
\emph{$x_{n}$-direction},  taking the naming convention as if the cube was
embedded in ${\mathbb{R}}^{n}$ with vertices $\{0, 1\}^{n}$.

The inductive hypothesis implies that there exists some direction,  let us call
it the \emph{$x_{1}$-direction},  in which all the edges in $Q^{ (0)}$ are of
the same color,  let us say the color \emph{red}. We wish to show that all the
edges in the $x_{1}$-direction in $Q^{ (1)}$ are also red. Since not all
vertical edges are of the same color,  there must exist some 3 dimensional
subcube $Q^{\prime}$ of $Q_{n}$,  which has edges in the $x_{1}$-direction,  the
vertical $x_{n}$-direction,  and some other third direction $x_{j}$,  such that
not all vertical edges in $Q^{\prime}$ are of the same color. The cube
$Q^{\prime}$ has the maximum,  7,  colors,  therefore one of its directions has all
edges of the same color. It cannot be the $x_{j}$-direction because the
$x_{1}$-direction bottom edges are red,  so we cannot have the two bottom
$x_{j}$-direction edges also of the same color by Theorem~\ref{main}~\ref{main4}. (we would
have a face with only 2 colors on a maximally colored 3-cube). Our choice of
$Q^{\prime}$ implies that it is not the vertical $x_{n}$-direction that has
all the same color. Hence all the $x_{1}$-direction edges in $Q^{\prime}$ are
of the same color,  and so they are all red.

Next pick an \textquotedblleft adjacent\textquotedblright\ cube $Q^{\prime
\prime}$ with edges in the $x_{1}$-direction,  $x_{n}$-direction and $x_{k}%
$-direction for some $k$,  such that $Q^{\prime\prime}$ and $Q^{\prime}$ share
an $ (x_{1}, x_{n})$-face. The two bottom edges in the $x_{1}$-direction in
$Q^{\prime\prime}$ are red,  and also the two edges in the $x_{1}%
$-direction on the face it shares with $Q^{\prime}$ are red. So $Q^{\prime
\prime}$ has at least 3 red edges in the $x_{1}$-direction,  and as it is
 colored with the maximum,  7,  colors,  all edges in the $x_{1}$-direction in $Q^{\prime\prime}$ are
red. We repeat this procedure until we have shown that all edges in the
$x_{1}$-direction in the top cube $Q^{ (1)}$ are red completing the proof.
\end{proof}

At this point we see that up to permuting the components of $Q_n$ (and then putting them back in order of the algorithm), Lemma~\ref{algoritheorem} yields all colorings with a maximum number of colors. Thus in our description of the set of all UOB as a bouquet of products of $\mathbb{P}^1$ given by the maps $\Phi_C$ for an admissible coloring of $Q_n$ the components of highest dimension ($2^{n+1}-2$) are described up to permutation of factors and order as the images of  $\Phi_C$ with $C$ given by the algorithm. Thus we have
\begin{theorem} \label{thmonecolor}
The irreducible components of maximum dimension of the variety of UOB are up to permutation of factors the images of $\Phi_C$  with $C$ given by the algorithm in Lemma~\ref{algoritheorem}.  In fact,  after
reordering factors we can write such a component as
\begin{equation*}
\mathcal{B} = 
\{
[a] \otimes \mathcal{B}_1 ,  [\widehat{a}] \otimes \mathcal{B}_2
\} , 
\end{equation*}
where $\mathcal{B}_i, i=1, 2$ are images of $\Phi_{C_i, i=1, 2}$ respectively with $C_1, C_2$ colorings of $Q_{n-1}$ given by the algorithm in Lemma~\ref{algoritheorem}.
\end{theorem}

\section{Distinguishability by local operations and classical communication} \label{sec:loccdist}
We now consider  LOCC distinguishability of elements of an $n$--qubit UOB. We are given  an unknown $n$--qubit state in a UOB,   and allowed a  protocol in which we can perform a  sequence of   local operations, that is, unitary transformations and local measurements on  qubits,  where the  choice of which qubit to measure at each step can  depend on the outcomes of the previous measurements (classical communication). We ask if this LOCC information  can determine with certainty which basis element was presented. Let us  consider  the two families of UOB in three
qubits corresponding to the first two displayed colorings above. The first is an example of a
coloring,  $C$,  with the maximum,  $7$,  colors. We consider the corresponding
bases,   of the form $\Phi_{C} ([u_{1}], \ldots, [u_{7}])$,  as in eq.~\eqref{basisvec},   and look at the basis state
$[u_{3}]\otimes\lbrack\hat{u}_{2}]\otimes\lbrack u_{4}]$. We note that if the first measurement
is in the first qubit (after applying the local unitary transformation taking $\left\vert {0}\right\rangle $ and $\left\vert {1}\right\rangle $  to $[u_{3}]$ and  $[\hat{u}_{3}]$ respectively),  then the outcome is $[u_{3}]$ with certainty.  From a second measurement in the second qubit (after applying the local unitary transformation taking $\left\vert {0}\right\rangle $ and $\left\vert {1}\right\rangle $  to $[u_{2}]$ and  $[\hat{u}_{2}]$ respectively), the outcome is 
 $[\hat{u}_{2}]$ with certainty. Similarly the measurement in the third qubit must be $[u_{4}]$  with certainty. We therefore have the correct state with certainty. Notice that the order of measurement is critical. We now consider the second example which is a maximal coloring of $Q_3$ using 6 colors. This example appears in~\cite{bdf:qnwe},  where it is shown that there is no ordered set of local transformations and measurements for the UOB of the form 
$\Phi_{C} (u_{1}, \ldots, u_{6})$,   with  $[u_i]\neq[ u_j]$,  that will determine a basis element with certainty.

Theorem~\ref{thmonecolor} implies that the discussion above for $\Phi_{C} (u_{1}
, \ldots, u_{2^{n}-1})$ for an admissible coloring of $Q_{n}$ with $2^n-1$ colors will work as long as the order is adapted to the algorithm,  in Lemma~\ref{algoritheorem},  that is used to
construct the coloring. Theorem $1$ of Walgate and Hardy~\cite{wh:nadbs} now implies 
that if $C$ is a maximal coloring of $Q_{n}$ with $k<2^{n}-1$ colors then
there is no such ordered set  of measurements that will identify with certainty a
specific state in $\Phi_{C} (u_{1}, \ldots, u_{k})$,  if all of the $u_{i}$ that
appear in a given factor are distinct. 

Distinguishability by LOCC is  also called {\it local} distinguishability~\cite{wh:nadbs}. We formally define it in the spirit of~\cite{wh:nadbs}.

\begin{definition} A UOB is locally distinguishable if there exists an ordering of tensor factors $(1,\ldots,n)$, and a sequence of measurements on respective tensor factors $\{M_1,\ldots,M_n\}$ such that:
\begin{enumerate}
\item $M_i$ for $i>1$ is a function of the outcomes of previous measurement results $\{r_j\}_{j=1,\ldots,i-1}$   from respective measurements $\{M_j\}_{j=1,\ldots,i-1}$.
\item The results $(r_1,\ldots,r_n)$  identify the  basis element of the UOB  on which the measurement is performed.
\end{enumerate}
\end{definition}

We restate Theorem $1$ in~\cite{wh:nadbs}  (with slight notational change). In this theorem, ``going first" refers  to the party (tensor factor) performing the first  measurement. 
\begin{theorem}[Walgate and Hardy]\label{thmwh}
Alice and Bob share a quantum system $\mathbb{C}^2\otimes\mathbb{C}^n$: Alice has a qubit, and Bob an $n$-dimensional
system that may be entangled with that qubit.
If Alice goes first, a set of $l$ orthogonal states $\{\psi_i\}_{i=1\ldots l}$ is
exactly locally distinguishable if and only if there is an orthogonal basis $\{a,\hat a\}$ for Alice's qubit, and orthonormal sets,  $\{\eta_a^i\}_{i=1\ldots l}$ and  $\{\eta_{\hat a}^i\}_{i=1\ldots l}$, in  Bob's system $\mathbb{C}^n$, such that:
\begin{equation} \label{locdis}
\psi_i = a\otimes \eta_a^i + \hat{a} \otimes \eta_{\hat a}^i
\end{equation}
\end{theorem}

\begin{corollary}
A  UOB  is locally distinguishable if and only if it is from  the family of UOB with  maximal dimension.
\end{corollary}

\begin{proof} Let the UOB be $\mathcal{B}$. An element of $b \in \mathcal{B}$ only has  one term in the sum in~\eqref{locdis}, either 
$b =  a\otimes \eta_a^i$, or $b = \hat{a} \otimes \eta_{\hat a}^i$.

Assume  $\mathcal{B}$ is from  the  family of UOB with maximal dimension. Let us show local distinguishability of $\mathcal{B}$. By Theorem~\ref{thmonecolor},   the form of  $\mathcal{B}$ is 
\begin{equation*}
\mathcal{B} = 
\{
[a] \otimes \mathcal{B}_1 , [\widehat{a}] \otimes \mathcal{B}_2
\} ,
\end{equation*}
where $\mathcal{B}_1$ and $\mathcal{B}_2$ are from  the maximal dimensional family of   UOB in ${(\mathbb{C}^2)}^{\otimes (n-1)}$. By induction then, we can assume that $\mathcal{B}_1$ and $\mathcal{B}_2$ are locally distinguishable, i.e., the conclusion is true for  $n-1$. Then local distinguishability of $\mathcal{B}$ (for $n$) follows by  Theorem~\ref{thmwh}. 

For the converse, assume local distinguishability of the UOB $\mathcal{B}$. Then by Theorem~\ref{thmwh}, the form of  $\mathcal{B}$ is 
\begin{equation} \label{formB}
\mathcal{B} = 
\{
[a] \otimes \mathcal{B}_1 , [\widehat{a}] \otimes \mathcal{B}_2
\} ,
\end{equation}
where $[a]$ is in the factor measured first, and $\mathcal{B}_1$ and $\mathcal{B}_2$ are some  UOB in the factors, ${(\mathbb{C}^2)}^{\otimes (n-1)}$, measured afterwards . Local distinguishability of $\mathcal{B}$ implies that of $\mathcal{B}_1$ and $\mathcal{B}_2$. By induction, $\mathcal{B}_1$ and $\mathcal{B}_2$ are from the  the maximal dimensional  family of UOB. Then by  dimension count in~\eqref{formB}, $\mathcal{B}$ is also from the maximal dimensional family of UOB in ${(\mathbb{C}^2)}^{\otimes n}$.
\end{proof}

This can also be seen as a direct consequence of Theorem 6 in~\cite{W:Gleason} and  substantiates our claims.
A sightly stronger result on {\it asymptotic} distinguishability is obtained from~\cite{kkb:apdlocc}, which allows the parties to have infinite resources and arbitrarily long times in their LOCC protocol.  Proposition $2$ of~\cite{kkb:apdlocc} implies that even under asymptotic LOCC, perfect discrimination is only possible for the UOB that belong to the maximum dimensional family.

\section{Constructions of maximal UOB not distinguishable by LOCC} \label{sec:nonlocc}

By now,  we know that the only UOB that are LOCC distinguishable belong to the maximum dimensional family, and we know which UOB belong to this family. Next we turn to constructions of maximal UOB for $n$ qubits, that are not from the maximum dimensional family. Such UOB give us families that are not distinguishable by LOCC, and therefore are the ones most useful in secure communication protocols like QKD. 

Equivalently we are looking for
maximally colored cubes with less than
the maximum number of colors.  We saw the
maximally colored $Q_3$ with 6 colors.  Let us construct an analogous
coloring on $Q_n$ for $n \geq 4$.

First, we color $Q_n$ 
with only two colors,
and call them `dominant' and `non-dominant'.  We color according to the rule
that every 2-face has to have 3 edges `dominant' and 1 edge `non-dominant'.
It is not hard to prove that once we pick a single `non-dominant' edge,
then the coloring of an $n$-cube
is forced up to mirror symmetry.  Each direction in the 3-cube
has 1 `non-dominant' and 3 `dominant' edges.  In the 4-cube, each
direction has 2 `non-dominant' and 6 `dominant' edges, and this process
can be continued for higher $n$.
We now replace the `dominant' color
with $n$ distinct colors, one for each direction.  The edges 
previously colored with `non-dominant', we color each with a distinct color,
and it can be checked the resulting coloring is admissible.
We obtain a maximal coloring, which can be shown just by considering
the 2-faces: Changing a proper subset of the `dominant' colors in a single
direction to a new color would break the admissibility condition for some
2-face.
We obtain $n(2^{n-3}+1)$ colors, which is less than the maximal number
of colors possible.  Therefore we obtain a family of UOB not distinguishable by LOCC
for every $n$.
Following the procedure for $Q_4$ we obtain the a UOB of the form:
\begin{equation*}
\begin{aligned}
&
[u_{4}]\otimes[u_{3}]\otimes[u_{2}]\otimes[u_{1}], & &
[u_{4}]\otimes[u_{6}]\otimes[u_{5}]\otimes[\hat{u}_{1}], 
\\
&
[u_{4}]\otimes[u_{3}]\otimes[\hat{u}_{2}]\otimes[u_{1}], & &
[u_{7}]\otimes[u_{3}]\otimes[\hat{u}_{5}]\otimes[\hat{u}_{1}], 
\\
&
[u_{8}]\otimes[\hat{u}_{3}]\otimes[ u_{5}]\otimes[u_{1}], & &
[u_{4}]\otimes[\hat{u}_{6}]\otimes[ u_{5}]\otimes[\hat{u}_{1}], 
\\
&
[u_{4}]\otimes[\hat{u}_{3}]\otimes[\hat{u}_{5}]\otimes[u_{9}], & &
[u_{4}]\otimes[\hat{u}_{3}]\otimes[\hat{u}_{5}]\otimes[\hat{u}_{9}],
\\
&
[\hat{u}_{4}]\otimes[u_{3}]\otimes[u_{5}]\otimes[u_{10}], & &
[\hat{u}_{4}]\otimes[u_{3}]\otimes[u_{5}]\otimes[\hat{u}_{10}], 
\\
&
[\hat{u}_{4}]\otimes[u_{11}]\otimes[\hat{u}_{5}]\otimes[u_{1}], & &
[\hat{u}_{7}]\otimes[u_{3}]\otimes[\hat{u}_{5}]\otimes[\hat{u}_{1}], 
\\
&
[\hat{u}_{8}]\otimes[\hat{u}_{3}]\otimes[ u_{5}]\otimes[u_{1}], & &
[\hat{u}_{4}]\otimes[\hat{u}_{3}]\otimes[ u_{12}]\otimes[\hat{u}_{1}], 
\\
&
[\hat{u}_{4}]\otimes[\hat{u}_{11}]\otimes[\hat{u}_{5}]\otimes[u_{1}], & &
[\hat{u}_{4}]\otimes[\hat{u}_{3}]\otimes[\hat{u}_{12}]\otimes[\hat{u}_{1}].
\end{aligned}
\end{equation*}
This $4$-qubit UOB is similar to the $3$-qubit example in FIG.~\ref{fig2}; it has $3$ colors in each direction distributed so that there are $6$ edges of one color and $1$ edge each of the other $2$ colors.

We can, in fact, construct   a large supply of  maximal families.
Let us start with a generalization of the construction we already used
to construct the maximal dimensional component.
Start with two UOBs
$\{ b_1,\ldots,b_{N} \}$ and $\{ c_1,\ldots,c_{N} \}$ where $N=2^{n-1}$,
with 
$m$ and $k$ distinct vectors (colors) respectively.  Let
$a$  be any unit vector in $\mathbb{C}^2$, and construct the UOB
\begin{equation*}
a \otimes b_1,
\ldots ,
a \otimes b_N,
\hat{a} \otimes c_1,
\ldots ,
\hat{a} \otimes c_N.
\end{equation*}
This UOB uses $m+k+1$ distinct vectors (colors).  If we start with at least
one of the UOBs being not LOCC distinguishable, that is, not part of the
maximal dimensional family, we again obtain a non-distinguishable family.

In terms of cubes, the above construction colors the $Q_n$ so that one
direction has a unique color.  Conversely it is not hard to see that if
one direction has a unique color, the two $Q_{n-1}$ which this direction
separates are then colored with distinct colors if the coloring is to be
maximal.  

We can also reverse the idea.  Instead of making the new factor use only
one vector, we can also use as many distinct vectors as possible in the new
factor.  Take a single UOB
$\{ b_1,\ldots,b_{N} \}$ with $m$ distinct vectors, $N=2^{n-1}$.  Then take
$N$ distinct vectors $a_1,\ldots,a_N \in \mathbb{C}^2$ and construct a new UOB
\begin{equation*}
a_1 \otimes b_1,
\ldots ,
a_N \otimes b_N,
\hat{a}_1 \otimes b_1,
\ldots ,
\hat{a}_N \otimes b_N.
\end{equation*}
The number of distinct vectors used is then $m+N = m+2^{n-1}$.
Again, if we start with a UOB not in the maximal dimensional component
we again obtain a nondistinguishable UOB.

%

As a remark, one may 
ask for the minimal number of colors in a maximally colored $Q_n$.
That is, the dimension of the lowest dimensional component of
UOB.  Let us call this number $C(n)$.  Using the constructions above and an
induction argument we leave it to the reader to prove:
\begin{equation*}
\begin{gathered}
C(2) = 3, \qquad C(3)=6, \\
2n \leq C(n) \leq 13(2^{n-4})-1 \quad \text{(if $n \geq 4$)} .
\end{gathered}
\end{equation*}

\section{Discussion}\label{sec:disc}

In this paper, we  presented several ideas and results  pertaining UOB for systems of $n$ qubits. To systematize our search for UOB, we began with  drawing a connection between UOB and colorings of an $n$-dimensional hypercube. This led us to the definition of an {\it admissible} coloring and a  partial order on such colorings.  The maximal elements of this order define families of UOB of dimensions corresponding to their number of colors.   Each coloring defines  a forest of colors, such that the maximum dimensional family corresponds to a single tree of $2^n -1$ colors (dimension of the family). This  gave us a complete characterization of the maximum dimensional family, and its structure. Knowing the structure it is apparent through a result of Walgate and Hardy~\cite{wh:nadbs} that the only LOCC distinguishable UOB belong to this family. 

From the perspective of secure communication, like the QKD protocols, it is the UOB that are {\it not} LOCC distinguishable that exhibit the nonlocality requisite in the success of the  protocols. The generic UOB being LOCC distinguishable, we constructed examples of maximal families of UOB   of dimensions less that the maximum. We generalized the earliest examples of such UOB (for $n=3$) in~\cite{bdf:qnwe} to arbitrary number $n$ of qubits,  and described other constructions that build maximal families from known ones. 

This leaves open certain immediate questions.  A complete characterization of all the families of UOB is the strongest of them. Short of that, it would be interesting to know what is the lower bound on the dimension of a maximal UOB. In the domain of applications, perhaps  more interesting secure communication protocols may be possible by employing these results. It would be very useful if the ideas we have presented 
could be extended directly as  tools to analyze UOB  for systems of qudits. Unfortunately given $[u] \in {\mathbb P}^{d-1}$  there is not a unique orthogonal
$[\hat{u}] \in {\mathbb P}^{d-1}$ if $d>2$. To address this ambiguity would require encoding further structure. 

\begin{acknowledgments}
The authors thank Gilad Gour for pointing out the work of~\cite{bdf:qnwe} and~\cite{scj:lsmco} on nonlocality without entanglement and for his patient explanation of LOCC to the third named  author. They also thank Nathaniel Johnston for his comments, particularly for making them aware of asymptotic distinguishability from~\cite{kkb:apdlocc}.
\end{acknowledgments}

%


\end{document}